\documentclass{IEEEtran}
\usepackage{cite}
\usepackage[pdftex]{graphicx}
\usepackage[cmex10]{amsmath}
\usepackage{amssymb}
\usepackage{amsthm}
\usepackage{algorithm}
\usepackage{algorithmic}
\usepackage[caption=false,font=footnotesize]{subfig}
\usepackage{url}
\usepackage{multirow}

\DeclareMathOperator{\real}{Re}
\DeclareMathOperator{\imag}{Im}

\newtheorem{thm}{Theorem}

\newtheorem{lem}[thm]{Lemma}

\title{Complex and Quaternionic Principal Component Pursuit and Its Application to Audio Separation}

\author{Tak-Shing~T.~Chan,~\IEEEmembership{Member,~IEEE}  and Yi-Hsuan~Yang,~\IEEEmembership{Member,~IEEE}%
\thanks{Manuscript received October 30, 2015; revised January 4, 2016; accepted Month xx, xxxx. Date of publication Month xx, xxxx; date of current version January 4, 2016. This work was supported by the Academia Career Development Program. The associate editor coordinating the review of this manuscript and approving it for publication was Prof.~Eric Moreau.}%
\thanks{The authors are with the Research Center for Information Technology Innovation, Academia Sinica, Taipei 11564, Taiwan (e-mail: takshingchan@citi.sinica.edu.tw; yang@citi.sinica.edu.tw).}%
\thanks{Digital Object Identifier 10.1109/LSP.2016.2514845}}

\IEEEpubid{\copyright~2016 IEEE. Personal use is permitted. For any other purposes, permission must be obtained from the IEEE by emailing pubs-permissions@ieee.org.}

\begin{document}

\maketitle

\begin{abstract}
Recently, the principal component pursuit has received increasing attention in signal processing research ranging from source separation to video surveillance. So far, all existing formulations are real-valued and lack the concept of phase, which is inherent in inputs such as complex spectrograms or color images. Thus, in this letter, we extend principal component pursuit to the complex and quaternionic cases to account for the missing phase information. Specifically, we present both complex and quaternionic proximity operators for the $\ell_1$- and trace-norm regularizers. These operators can be used in conjunction with proximal minimization methods such as the inexact augmented Lagrange multiplier algorithm. The new algorithms are then applied to the singing voice separation problem, which aims to separate the singing voice from the instrumental accompaniment. Results on the iKala and MSD100 datasets confirmed the usefulness of phase information in principal component pursuit.
\end{abstract}


\section{Introduction}
\label{sec:intro}

\IEEEPARstart{I}n recent years, robust principal component analysis (RPCA) \cite{Candes11} has been quite successful in various signal processing applications including source separation, face recognition, and video surveillance \cite{Huang12,Ikemiya15,Peng10,Bouwmans14}. RPCA works by decomposing an input matrix $X\in\mathbb{R}^{m \times n}$ into a low-rank matrix $A$ plus a sparse matrix $E$:
\begin{equation}
\label{eq:1}
\min_{A,E}\mathrm{rank}(A)+\lambda\|E\|_0\mbox{\quad s.t.\quad}X = A+E.
\end{equation}
Unfortunately, the above formulation is NP-hard. Hence, the principal component pursuit (PCP) \cite{Candes11} instead solves the following relaxed problem:
\begin{equation}
\label{eq:2}
\min_{A,E}\|A\|_*+\lambda\|E\|_1\mbox{\quad s.t.\quad}X = A+E\,,
\end{equation}
where $\|\cdot\|_*$ is the trace norm (sum of singular values), $\|\cdot\|_1$ is the entrywise $\ell_1$-norm, $\lambda$ is a positive parameter which is set to $k/\sqrt{\max(m,n)}$, and $k$ denotes the trade-off between the rank of $A$ and the sparsity of $E$ \cite{Candes11,Huang12}. Under weak conditions and $k=1$, it has been proven that PCP has a high probability to exactly recover the low-rank and sparse components \cite{Candes11}, although $k$ can be adjusted if the conditions are violated. Most implementations of PCP are based on proximal minimization \cite{Combettes11} which is an extension of gradient projection in the nondifferentiable case. The proximity operator of a function $f:\mathbb{R}^p\rightarrow\mathbb{R}^p$ is defined as \cite{Combettes11}
\begin{equation}
\label{eq:21}
\mathrm{prox}_f\mathbf{z}=\arg\min_\mathbf{x}\left(\frac{1}{2}\|\mathbf{z}-\mathbf{x}\|_2^2+f(\mathbf{x})\right),\ \mathbf{x}\in\mathbb{R}^p,
\end{equation}
with closed-form solutions such as the soft-thresholding \cite{Donoho95} and singular value thresholding \cite{Cai10} operators for the $\ell_1$- and trace-norm regularizers, respectively. The resulting PCP algorithm in \cite{Candes11} is based on the well-known inexact augmented Lagrange multiplier algorithm (IALM), which has good convergence guarantees \cite{Lin09}. Their algorithm looks exactly like Algorithm~\ref{alg:1} below, except that the input  matrix $X$ is real.

\IEEEpubidadjcol

\subsection{Related Work}
\label{subsec:related}

The objective of the singing voice separation (SVS) problem is to separate the singing voice component from an audio mixture containing both the singing voice and the instrumental accompaniment. First proposed in \cite{Huang12}, PCP-SVS \cite{Huang12,Yang12,Chan15,Ikemiya15} assumes that the magnitude spectrogram of pop music can be decomposed via \eqref{eq:2} into a low-rank instrumental component $A$ and a sparse voice component $E$. This assumption is based on the premise that the instrumental accompaniment is usually repetitive (hence low-rank), while the vocalist can only sing one note at a time (hence sparse). Then, the separated components are reconstructed using overlap-add with the original phases in the mixture (see Fig.~\ref{fig:pcp-svs}). As PCP-SVS decomposes entire spectrograms instead of individual frames, it is able to exploit statistical redundancies at both the local and global time scales. This approach assumes that the magnitude spectrograms are additive; however, prior to the invention of PCP-SVS, King and Atlas \cite{King10} has already demonstrated that magnitude additivity does not hold when the phases differ. Furthermore, research in parametric spatial audio \cite{Traa14} suggests that inter-channel (stereo) phase might also be important. Motivated by these observations, we aim to extend PCP to the complex and quaternion domains. More specifically, by solving for relevant proximity operators in these domains, the extended PCP will be able to preserve not only the spectral phase but inter-channel phase as well. We hypothesize that the preserved phase will improve the performance of signal processing applications such as SVS.

Although there have been some work on quaternion PCA \cite{Pei03,LeBihan03}, a quaternion version of RPCA has not been established. An implementation of the quaternion singular value decomposition (SVD), based on real bidiagonalization using quaternion Householder transformations, is available in the Quaternion Toolbox for Matlab (QTFM) \cite{10.1016/j.amc.2006.04.032}. However, as this \textsc{Matlab} implementation is inefficient, we will use an older but faster algorithm \cite{Zhang97,Pei03,LeBihan04} throughout the paper.

Our contributions in this paper are twofold. First, we will extend PCP to the complex and quaternion domains (which are phase-preserving) with some quaternion algebra. Second, we will test their performances on two audio source separation competition datasets, ascertaining their usefulness.

This paper is organized as follows. In Section \ref{sec:prelim}, we recall without proof some basic facts about quaternion matrices. In Section \ref{sec:pcp}, we present the complex and quaternionic PCP. Then, we describe our experiments using real, complex and quaternionic PCP on the iKala and Mixing Secrets datasets in Section \ref{sec:exp} and conclude in Section \ref{sec:conc}.

\begin{figure}[!t]
\centering
\subfloat[]{
\includegraphics[width=\columnwidth]{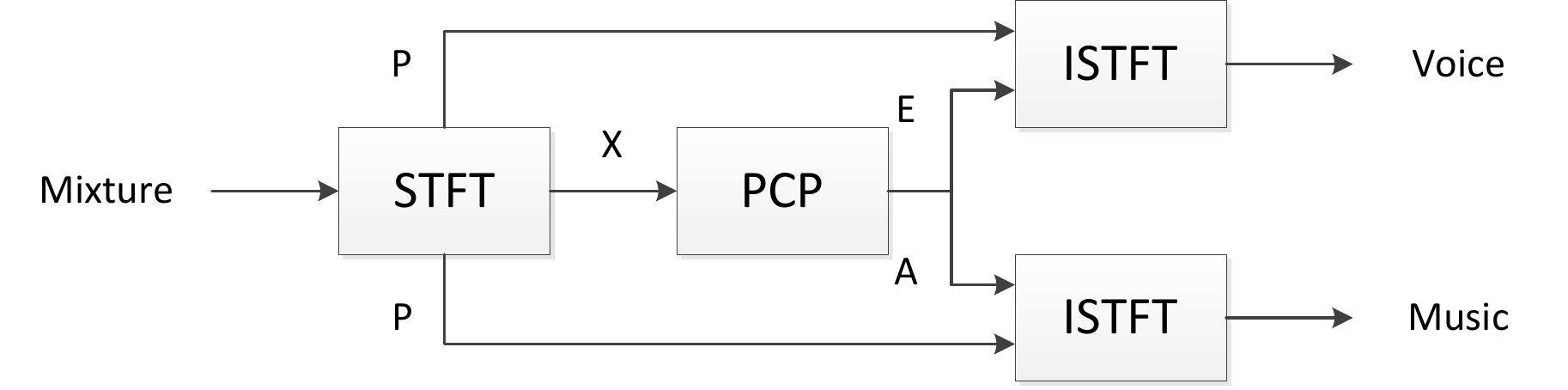}
}\\
\subfloat[]{
\includegraphics[width=\columnwidth]{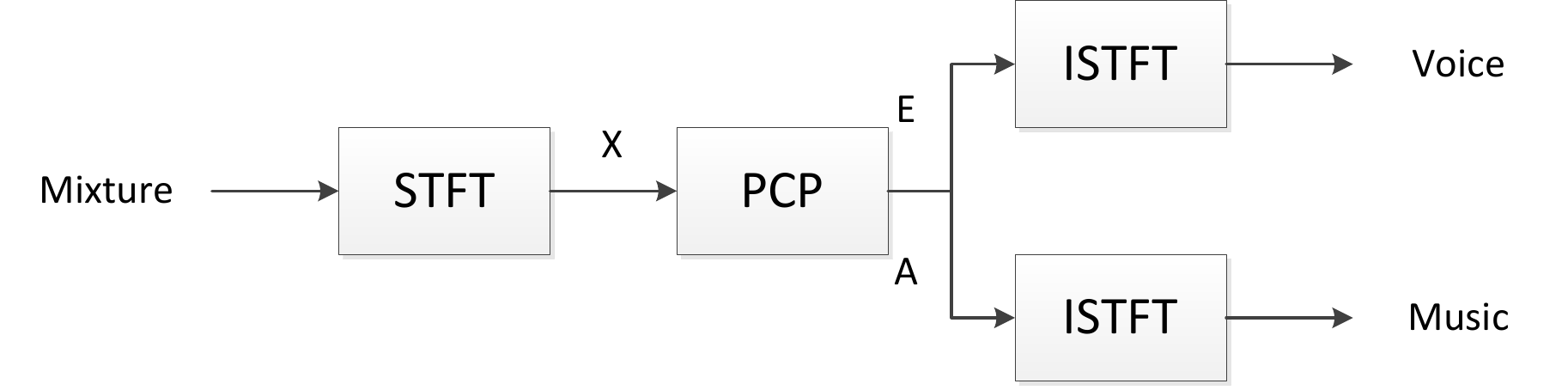}
}
\caption{Block diagram of PCP-SVS systems. Refer to \eqref{eq:2} for the meaning of $X$, $E$, and $A$. (a) In real PCP, $X$ contains the magnitude only; the phase $P$ is lost and has to be copied from the original mixture for ISTFT. (b) In complex and quaternionic PCP, the phases are preserved. For the quaternionic case only, the STFT and ISTFT blocks multiplex and demultiplex the stereo spectrograms to and from a quaternionic spectrogram (see Section \ref{sec:exp}).}
\label{fig:pcp-svs}
\end{figure}

\section{Preliminaries}
\label{sec:prelim}

The quaternions $\mathbb{H}$ is a superset of the complex numbers $\mathbb{C}$ with four dimensions instead of two, i.e., $q=a_0+a_1\imath+a_2\jmath+a_3\kappa$, where $a_0,a_1,a_2,a_3\in\mathbb R$, with imaginary units $\imath,\jmath,\kappa$ such that $\imath^2=\jmath^2=\kappa^2=\imath\jmath\kappa=-1$ \cite{Zhang97}. Here $\real(q)=a_0$ is called the real part and $\imag(q)=a_1\imath+a_2\jmath+a_3\kappa$ is called the imaginary part. If $\real(q)=0$, $q$ is called a pure quaternion. The quaternion conjugate and magnitude are defined as $\bar q=a_0-a_1\imath-a_2\jmath-a_3\kappa$ and $|q|=\sqrt{a_0^2+a_1^2+a_2^2+a_3^2}$, respectively. A quaternion can be uniquely represented as a pair of complex numbers \cite{Zhang97}: if $x=a_0+a_1\imath$ and $y=a_2 +a_3\imath$, then $q=x+y\jmath=a_0+a_1\imath+a_2\jmath+a_3\kappa$ and vice versa. Thus the complex numbers are indeed a subset of the quaternions.

\paragraph{Complex Matrix Isomorphism}

For any quaternionic matrix $A\in\mathbb{H}^{m\times n}$, there is a well-known complex isomorphism $\chi:\mathbb{H}^{m\times n}\rightarrow\mathbb{C}^{2m\times 2n}$, defined by \cite{Zhang97}:
\begin{equation}
\label{eq:9}
\chi(A)=\chi({X+Y\jmath})=\left[\begin{array}{cc}
X&Y\\
-\overline{Y}&\overline{X}
\end{array}\right],
\end{equation}
where $X,Y\in\mathbb{C}^{m\times n}$ is the unique representation of $A$ such that $A=X+Y\jmath$. This isomorphism has the properties that $\chi(AB)=\chi(A)\chi(B)$, $\chi(A^*)=\chi(A)^*$, and $\mathrm{tr}(\chi(A))=2\real\mathrm{tr}(A)$ for all $A,B\in\mathbb{H}^{m\times n}$ \cite{Knapp02}. The truncated SVD of $A$ can also be performed on $\chi(A)$ directly, where the singular values are the same as those of $A$, except that they occur in pairs \cite{Pei03}. This isomorphism allows us to simplify our proofs by working in an isomorphic complex domain.

\paragraph{Real Vector Isomorphism}

For $A,B\in\mathbb{H}^{m\times n}$, it has long been known that $\real\mathrm{tr}(AB^*)$ is isomorphic to the Euclidean inner product on $\mathbb{R}^{4mn}$ \cite{Tai68}. In particular, we can first transform $A$ into a real matrix in $\mathbb{R}^{m\times 4n}$ by \cite{Kumar12}
\begin{equation}
\label{eq:15}
[\real(A),\imag_\imath(A),\imag_\jmath(A),\imag_\kappa(A)],
\end{equation}
then further vectorize the results into a real vector in $\mathbb{R}^{4mn}$ (likewise for $B$). According to \cite{Tai68}, their dot product in $\mathbb{R}^{4mn}$ is equivalent to $\real\mathrm{tr}(AB^*)$ given the original quaternionic matrices. So it makes sense to define the quaternionic inner product as $\langle A,B\rangle=\real\mathrm{tr}(AB^*)=\real\mathrm{tr}(A^*B)$ \cite{Knapp02}, which is nonstandard but obeys all the real inner product axioms due to the aforementioned equivalence. Furthermore, its induced quaternionic Frobenius norm $\|A\|_F=\sqrt{\langle A,A\rangle}$ satisfies \cite{LeBihan04}:
\begin{equation}
\label{eq:19}
\|A\|_F=\sqrt{\sum_i\sigma_i^2(A)}=\sqrt{\frac{1}{2}\sum_i\sigma_i^2(\chi(A))},
\end{equation}
where $\sigma_i(\cdot)$ denotes the singular values in any order.

\section{Complex and Quaternionic PCP}
\label{sec:pcp}

In this section, we will extend the real PCP to the complex and quaternionic cases. As the complex numbers are a subset of the quaternions, we only need to prove the quaternion case.

\subsection{Derivation of Proximity Operators}
\label{subsec:prox}

We begin by extending the proximal operator itself.

\begin{thm}
\label{thm:3.0}
The proxmity operator \eqref{eq:21} can be extended to the quaternion and complex cases via:
\begin{equation}
\label{eq:21b}
\mathrm{prox}_f\mathbf{z}=\arg\min_\mathbf{x}\left(\frac{1}{2}\|\mathbf{z}-\mathbf{x}\|_2^2+f(\mathbf{x})\right),\ \mathbf{x}\in F^p,
\end{equation}
where $F$ is $\mathbb{H}$ or $\mathbb{C}$.
\end{thm}

\begin{proof}
One approach is to transform the quaternionic vectors into real vectors, then invoke \eqref{eq:21} after compensating for any possible differences inside $f(\mathbf{x})$. We can use the real isomorphism from the vectorization of \eqref{eq:15} for this. Due to the definition of the quaternion magnitude, $\frac{1}{2}\|\mathbf{z}-\mathbf{x}\|_2^2$ is invariant under this transformation, so we can (and will) equivalently extend the domain of $f$ to $\mathbb{H}^p$ without needing to adjust $f(\mathbf{x})$ in what follows. This completes the proof.
\end{proof}

We now treat the $\ell_1$- and trace-norm regularizers in turn.

\begin{thm}
\label{thm:3.2}
The proximity operator for the quaternionic and complex $\ell_1$ regularizers $\lambda\|X\|_1$, where the $\ell_1$-norm operates entrywise, is:
\begin{equation}
\label{eq:25}
\mathrm{prox}_{\lambda\|\cdot\|_1}^F\mathbf{z}=\left(1-\frac{\lambda}{|\mathbf{z}|}\right)_+\mathbf{z},\ \mathbf{z}\in F^p,
\end{equation}
where $F$ is $\mathbb{H}$ or $\mathbb{C}$ and $\mathbf{z}=\mathrm{vec}\ Z$.
\end{thm}

\begin{proof}
It is a known result \cite{Kumar12} that the quaternionic lasso
\begin{equation}
\label{eq:23}
\min_\mathbf{a}\frac{1}{2}\|\mathbf{x}-D\mathbf{a}\|_2^2+\lambda\|\mathbf{a}\|_1,\ \mathbf{x}\in\mathbb{H}^m,\mathbf{a}\in\mathbb{H}^n,
\end{equation}
with $D\in\mathbb{R}^{m\times n}$, is equivalent to the group lasso \cite{Yuan06}
\begin{equation}
\label{eq:24}
\min_A\frac{1}{2}\|X-DA\|_F^2+\lambda\|A\|_{1,2},\ X\in\mathbb{R}^{m\times 4},\ A\in\mathbb{R}^{n\times 4}
\end{equation}
via the transformation in \eqref{eq:15}.
By setting $D=I$ in (\ref{eq:23}--\ref{eq:24}) and assigning each quaternion vector element to its own group, we get from \cite{Yuan06} the required proximity operator for the quaternionic $\ell_1$-norm regularizer.
\end{proof}

Perhaps not surprisingly, \eqref{eq:25} looks exactly the same as the soft-thresholding operator \cite{Donoho95} which is the corresponding operator in the real case. Note that the complex and quaternionic soft-thresholding operators already exist \cite{Sardy00,Jin12}, but they are not solved in the proximal form above. More recently, the proximity operator for the complex $\ell_1$-norm has been solved \cite{Maleki13}, but the quaternionic case remains open until now. Next we will deal with the trace-norm regularizer by generalizing both the von Neumann trace inequality \cite{Horn13} and a proof in \cite{Tomioka12} to the quaternionic case.

\begin{lem}
\label{lem:3}
For any two compatible quaternionic matrices, the von Neumann trace inequality also holds:
\begin{equation}
\label{eq:26}
\real\mathrm{tr}(A^*B)\leq\sum_i \sigma_i(A)\sigma_i(B),
\end{equation}
where the singular values $\sigma_i(\cdot)$ are in a nonincreasing order.
\end{lem}

\begin{proof}
By the properties of the complex matrix isomorphism \eqref{eq:9}, we have the following:
\begin{align}
\begin{split}
\real\mathrm{tr}(A^*B)=&\frac{1}{2}\mathrm{tr}(\chi(A^*B))\\
=&\frac{1}{2}\mathrm{tr}(\chi(A)^*\chi(B))\\
\leq&\frac{1}{2}\sum_i \sigma_i(\chi(A))\sigma_i(\chi(B)).
\end{split}
\end{align}
The last line is due to the original von Neumann trace inequality \cite{Horn13}. Since $\chi(A)$ outputs each singular value twice, we have $\real\mathrm{tr}(A^*B)\leq\sum_i \sigma_i(A)\sigma_i(B)$ in the quaternionic case too. To our best knowledge, this result is new.
\end{proof}

\begin{thm}
\label{thm:3.6}
The proximity operator for the trace-norm regularizer $\lambda\sum_i\sigma_i(X)$ is:
\begin{equation}
\label{eq:27}
\mathrm{prox}_{\lambda\|\cdot\|_*}\mathbf{z}=\mathrm{vec}\ U(\Sigma-\lambda)_+V^*,\ \mathbf{z}\in F^p,
\end{equation}
where $\mathbf{z}=\mathrm{vec}\ Z$, $U\Sigma V^*$ is the SVD of $Z$ with singular values $\Sigma_{ii}=\sigma_i(Z)$ in a nonincreasing order, and $F$ is $\mathbb{H}$ or $\mathbb{C}$.
\end{thm}

\begin{proof}
The real case has been proven in \cite{Tomioka12}. This proof is virtually identical to the real case except that we are additionally endowed with Lemma \ref{lem:3}, which allows us to extend the results to the complex and quaternionic cases. By \eqref{eq:19}, and the Euclidean inner product identity $\langle z-x,z-x\rangle=\langle z,z\rangle -2\langle z,x\rangle +\langle x,x\rangle$, which is valid because of the real vector isomorphism, we can deduce that:
\begin{align}
\label{eq:28}
\begin{split}
&\left\Vert Z-X\right\Vert_F^2+\lambda\sum_i\sigma_i(X)\\
=&\sum_i\sigma_i^2(Z)-2\left\langle Z,X\right\rangle +\sum_i\sigma_i^2(X)+\lambda\sum_i\sigma_i(X)\\
\geq&\sum_i\sigma_i^2(Z)-2\sum_i\sigma_i(Z)\sigma_i(X)+\sum_i\sigma_i^2(X)+\lambda\sum_i\sigma_i(X)\\
=&\sum_i\left(\sigma_i(Z)-\sigma_i(X)\right)^2+\lambda\sum_i\sigma_i(X),
\end{split}
\end{align}
where Lemma \ref{lem:3} is invoked on the penultimate line. The last line can be seen as $\|\sigma(Z)-\sigma(X)\|_2^2+\lambda\|\sigma(X)\|_1$ which can be minimized by the soft-thresholding function (\ref{eq:30}).
\end{proof}

\subsection{The Extended PCP Formulation}
\label{subsec:form}

Finally, we define the complex and quaternionic PCP as:
\begin{equation}
\label{eq:29}
\min_{A,E}\|A\|_*+\lambda\|E\|_1\mbox{\quad s.t.\quad}X = A+E\,,
\end{equation}
where $X\in\mathbb{C}^{m\times n}$ for the complex PCP and  $X\in\mathbb{H}^{m\times n}$ for the quaternionic PCP. This can be solved by the same algorithms from \cite{Lin09}, except that the soft-thresholding function:
\begin{equation}
\label{eq:30}
\mathcal{S}_\lambda[x]=\left\{\begin{array}{ll}
x-\lambda,&\mbox{if }x>\lambda,\\
x+\lambda,&\mbox{if }x<-\lambda,\\
0,&\text{otherwise}
\end{array}\right.
\end{equation}
should be changed to $\mathrm{prox}_{\lambda\|\cdot\|_1}^\mathbb{H}\mathbf{z}$ and $\mathrm{prox}_{\lambda\|\cdot\|_1}^\mathbb{C}\mathbf{z}$ for the quaternionic and complex PCP, respectively. The inexact augmented Lagrange multiplier (IALM) adaptation is shown in Algorithm \ref{alg:1}.

\begin{algorithm}[t]
\caption{Complex and Quaternionic PCP via IALM}
\begin{algorithmic}[1]
\label{alg:1}
\REQUIRE $X\in F^{m\times n}$, $F\in\{\mathbb{H},\mathbb{C}\}$, $\lambda\in\mathbb{R}$, $\mu\in\mathbb{R}^\infty$
\ENSURE $A_k$, $E_k$
\STATE Let $E_1=0$, $Y_1=X/\max\left(\|X\|_2,\lambda^{-1}\|X\|_\infty\right)$, $k=1$
\WHILE{not converged}
    \STATE $A_{k+1}\gets\mathrm{prox}_{1/\mu_k\|\cdot\|_*}(X-E_k+\mu_k^{-1}Y_k)$
    \STATE $E_{k+1}\gets\mathrm{prox}_{\lambda/\mu_k\|\cdot\|_1}^F(X-A_{k+1}+\mu_k^{-1}Y_k)$
    \STATE $Y_{k+1}\gets Y_k + \mu_k (X-A_{k+1}-E_{k+1})$
    \STATE $k\gets k+1$
\ENDWHILE
\end{algorithmic}
\end{algorithm}

\section{Experiments}
\label{sec:exp}

We will use the SVS task to compare the real, complex and quaternionic versions of PCP. Specifically, we will evaluate the effects of the following three levels of phase-informedness on source separation performance:
\begin{itemize}
\item Real PCP (no phases);
\item Complex PCP (spectral phase only);
\item Quaternionic PCP (spectral and inter-channel phases).
\end{itemize}

\subsection{Experimental Setup}
\label{subsec:setup}

Our evaluation employs two source separation competition datasets, the iKala \cite{Chan15} (MIREX) and MSD100\footnote{\url{http://corpus-search.nii.ac.jp/sisec/2015/MUS/MSD100_2.zip}} (SiSEC) datasets. The iKala dataset contains 252 30-second mono clips, whereas the MSD100 dataset contains 100 full stereo songs with durations ranging from 2'22'' to 7'10''. To reduce computations, we use only 30-second fragments (1'45'' to 2'15'') from each MSD100 song. The choice of this time period is informed by the fact that this is the only period where all 100 songs contain vocals. Evaluation is done with BSS Eval Version 3 \cite{Vincent12}, which calculates the source-to-distortion ratio (SDR), source-to-interference ratio (SIR), and source-to-artifact ratio (SAR) \cite{Vincent12} for both the instrumental ($A$) and vocal ($E$) parts. For stereo signals, we additionally have the source-image-to-spatial-distortion ratio (ISR). From SDR we calculate the normalized SDR (NSDR) \cite{Hsu10} by $\text{SDR}(\hat{v},v)-\text{SDR}(x,v)$, where $\hat{v}$ is the separated voice part, $v$ is the original clean voice, and $x$ is the original mixture. The NSDR for the instrumental part is calculated in the same manner. The NSDR can be interpreted as the improvement in SDR using the mixture itself as the baseline. Finally, we aggregate the performance over all clips by taking the weighted average, with weight proportional to the length of each clip \cite{Hsu10}. The resulting measures are denoted as GNSDR, GSDR, GISR, GSIR, and GSAR, respectively. For these measures, a larger value means better. GNSDR and GSDR are the most important as they measure the overall distortion \cite{Vincent12}.

Both datasets are downsampled from 44\,100 Hz to 22\,050 Hz to reduce memory usage. The singing voice and instrumental accompaniment are mixed at 0 dB signal-to-noise ratio. Our main setup is identical to Fig.~\ref{fig:pcp-svs}. We use a short-time Fourier transform (STFT) with a 1\,411-point Hann window with 75\% overlap as in \cite{Chan15}. For real and complex PCP, the two-channel stereo mixtures are further downmixed into a single mono channel. In the real case, the magnitude part is fed into PCP and the separated parts are reconstructed via inverse STFT using the original phase \cite{Huang12}; in the complex case, the complex spectrogram is fed directly into complex PCP and reconstructed without phase substitution. Finally, in the quaternionic case, the stereo signal is represented using the quaternion format $L+R\jmath$, where $L$ and $R$ contain the complex spectrograms for the left and right channels, respectively. The value of $k$ (i.e., the trade-off between the trace norm and the $\ell_1$-norm) is empirically determined to be 1.5 for the iKala dataset and 3 for the MSD100 dataset.\footnote{Our implementation is available at \url{http://mac.citi.sinica.edu.tw/ikala/code.html} which contains all the code to reproduce the results.}

\subsection{Results and Analysis}
\label{subsec:results}

Results for the iKala and MSD100 datasets are shown in Tables~\ref{tab:1} and \ref{tab:2}, respectively. Twenty-eight one-tailed paired $t$-tests are performed to determine whether Complex $>$ Real and Quaternion $>$ Complex (all with $p < 0.05$ after Bonferroni correction, except the daggered ones which are insignificant). For GNSDR and GSDR, complex PCP clearly outperformed real PCP on both datasets. Furthermore, for the instrumental part of the MSD100 dataset, quaternionic PCP performed better than its complex counterparts on all five measures. This means that, with the exception of quaternionic voice, the more phase-informed the better the separation. We can see that stereo phase is useful for the quaternionic instrumental part, where GISR significantly outperforms its complex counterpart, suggesting a superior spatial (stereo) reconstruction. The lack of performance in the quaternionic voice case is probably a drawback of the PCP formulation \eqref{eq:29}, where the $\ell_1$-norm is intrinsically phase-removing and we can only rely on the trace norm for phase preservation. Further work is required to improve this. However, judging from a noise removal perspective, this paper is already useful for singing voice removal applications.

\begin{table}[!t]
\renewcommand{\arraystretch}{1.1}
\caption{Results for iKala instrumental ($A$) and vocal ($E$), in dB}
\label{tab:1}
\centering
\begin{tabular}{|l|l|r@{.}l|r@{.}l|r@{.}l|r@{.}l|}
\hline
\multicolumn{2}{|l|}{} & \multicolumn{2}{|l|}{GNSDR} & \multicolumn{2}{|l|}{GSDR} & \multicolumn{2}{|l|}{GSIR} & \multicolumn{2}{|l|}{GSAR}\\
\hline\hline
\multirow{2}{*}{Real} & A & \phantom{-0}3&98 & \phantom{0}0&11 & 1&33 & 9&65\\
& E & 2&41 & 6&36 & 11&17 & 9&46\\
\hline
\multirow{2}{*}{Complex} & A & 5&46 & 1&59 & 3&64 & 8&33$^\dagger$\\
& E & 3&45 & 7&40 & 10&40$^\dagger$ & 12&10\\
\hline
\end{tabular}
\end{table}

\begin{table}[!t]
\renewcommand{\arraystretch}{1.1}
\caption{Results for MSD100 instrumental ($A$) and vocal ($E$), in dB}
\label{tab:2}
\centering
\begin{tabular}{|l|l|r@{.}l|r@{.}l|r@{.}l|r@{.}l|r@{.}l|}
\hline
\multicolumn{2}{|l|}{} & \multicolumn{2}{|l|}{GNSDR} & \multicolumn{2}{|l|}{GSDR} & \multicolumn{2}{|l|}{GISR} & \multicolumn{2}{|l|}{GSIR} & \multicolumn{2}{|l|}{GSAR}\\
\hline\hline
\multirow{2}{*}{Real} & A & \phantom{-0}3&57 & 8&92 & 13&18 & 10&78 & 22&13\\
& E & 3&11 & --1&41 & 3&88 & 6&74 & 0&28\\
\hline
\multirow{2}{*}{Complex} & A & 3&70 & 9&05 & 14&32 & 10&65$^\dagger$ & 23&03\\
& E & 3&30 & --1&23 & 2&82$^\dagger$ & 8&66 & 0&63\\
\hline
\multirow{2}{*}{Quaternion} & A & 5&00 & 10&35 & 18&91 & 10&71 & 23&25\\
& E & 3&15$^\dagger$ & --1&38$^\dagger$ & 2&75$^\dagger$ & 8&32$^\dagger$ & 0&57$^\dagger$\\
\hline
\end{tabular}
\end{table}

\section{Conclusions}
\label{sec:conc}

We have extended the PCP formation of RPCA, first by introducing the notion of complex and quaternionic proximity operators, then by adapting the proximity operators of the $\ell_1$- and trace-norm regularizers to the complex and quaternionic cases. Apart from the complex $\ell_1$-norm case \cite{Maleki13}, all of the proposed proximity operators are new. Our extensions are phase-preserving and can be used in a wide range of signal processing applications including audio source separation. Evaluation on the iKala and MSD100 datasets showed that the preserved phase information would increase SVS performance. Other PCP-SVS variants, such as RPCAh \cite{Yang12}, RPCA-F0 \cite{Ikemiya15}, and VD-RPCA \cite{Lehner15} are all real-valued so our extended formulation here can potentially improve their performance. We also expect the quaternionic PCP to work for color face recognition \cite{Pei03,Candes11}, because it is based on a noise removal paradigm so the $E$ part is irrelevant.

\section*{Acknowledgment}

The authors would like to thank the anonymous reviewers for their valuable comments on the presentation of this letter.

\IEEEtriggeratref{15}
\bibliographystyle{IEEEtran}
\bibliography{chan16spl}

\end{document}